\documentclass[11pt]{IEEEtran}
\usepackage{amsmath,amsfonts,xyling}

\newtheorem{prop}{Proposition}
\newtheorem{lem}[prop]{Lemma}
\newtheorem{cor}[prop]{Corollary}
\newtheorem{df}[prop]{Defintion}
\newtheorem{exa}[prop]{Example}
\newtheorem{rem}[prop]{Remark}

\newcommand{\etalchar}[1]{$^{#1}$}

\DeclareMathOperator{\Aut}{Aut}
\DeclareMathOperator{\End}{End}

\DeclareMathOperator{\ord}{ord}
\newcommand{\mcX}{\mathcal{X}}
\newcommand{\F}{\mathbb{F}}
\newcommand{\Z}{\mathbb{Z}}

\newcommand{\QAnode}[2]
{\K{\fbox{\tiny{$\begin{matrix}#1=\overline{0}\\ #2=\overline{2}\end{matrix}$}}}} 

\title{Efficient recovering of operation tables\\
  of black box groups and rings}

\author{
  Jens Zumbr\"agel, G\'erard Maze and Joachim Rosenthal\\
  Mathematics Institute\\ 
  University of Zurich\\
  Winterthurerstrasse 190\\
  CH - 8057 Zurich,  Switzerland\\ 
  \texttt{www.math.unizh.ch/aa}%
  \thanks{Authors were supported in part by Swiss
  National Science Foundation Grant no.\ 107887.}
}

\begin{document}

\maketitle 

\begin{abstract}
  People have been studying the following problem: Given a finite set
  $S$ with a hidden (black box) binary operation $*:S\times S\to S$
  which might come from a group law, and suppose you have access to an
  oracle that you can ask for the operation $x*y$ of single pairs
  $(x,y)\in S^2$ you choose. What is the minimal number of queries to
  the oracle until the whole binary operation is recovered, i.e. you
  know $x*y$ for all $x,y\in S$?
  
  This problem can trivially be solved by using $|S|^2$ queries to the
  oracle, so the question arises under which circumstances you can
  succeed with a significantly smaller number of queries.
  
  In this presentation we give a lower bound on the number of queries
  needed for general binary operations. On the other hand, we present
  algorithms solving this problem by using $|S|$ queries, provided
  that $*$ is an abelian group operation.  We also investigate black
  box rings and give lower and upper bounds for the number of queries
  needed to solve product recovering in this case.
\end{abstract}


\section{Introduction}

There is a considerable literature on algebraic objects whose
operations are described by a `black box'. There are different
motivations for studying such objects.

In computational group theory, \emph{black box groups} became an
important and frequently used tool.  They were introduced by Babai and
Szemer{\'e}di~\cite{ba84b} in order to study algorithms for matrix
groups. In a black box group elements are encoded as (not necessarily
unique) bitstrings and there are oracles (the black box) providing
multiplication and inversion of the encoded group elements as well as
recognition of the identity element---but often the isomorphy class
and even the order of the underlying group is unknown.  The research
is ongoing and a collection of some algorithms for black box groups
can be found in Seress' monograph~\cite[Ch. 2]{se03}.  One major
question in black box group research is the recognition problem which
asks whether a given black box group is isomorphic to a fixed finite
group like $S_n$ or $SL_n(\F_q)$ and possibly to provide an explicit
description of such an isomorphism, see e.g.~\cite{be03,br00,ka01}.
The constructive recognition problem in the case of abelian groups has
been investigated by Buchmann, Jacobson, Schmidt and
Teske~\cite{bu97a,bu05}.

Also in cryptography the use of black box groups and fields proved
themselves useful when analyzing the hardness of the discrete
logarithm problem, e.g. Shoup obtained this way lower bounds for
generic algorithms~\cite{sh97b}.  Black box fields of prime order were
used by Boneh and Lipton~\cite{bo96} when studying the discrete
logarithm problem in a presence of a Diffie-Hellman oracle.  Note that
the isomorphy class of the underlying objects are known in these
cases.

The basic question we are interested in this paper is, given a black
box group or black box ring, how many calls to the oracle are
necessary to recover the whole operation tables. To illustrate the
importance of this question we mention its relevance for estimating
the information content of an algebraic operation table and for
designing practical compression algorithms for these tables.
Furthermore, a black box object might be a crucial device in a
symmetric cryptosystem and one wishes to analyze the cost to describe
this black box completely.

We shall be interested in the problem of recovering the hidden
operation by using a minimal number of queries to the oracle.  In
algorithm analysis we neglect here the remaining computational costs,
i.e. we assume unlimited computational and storage power but limited
access to the oracle.

When investigating the product recovering problem it is natural to
assume unique encoding of group elements, since for the black box
operation to depend only on the underlying group and not on some
encoding arbitrariness this is necessary.  However we include some
remarks and comments concerning the general case of nonunique
encoding.

The organization of this paper is as follows.  In Section~\ref{one},
which forms the main part, we consider the case of one binary black
box operation. After a formalization of the problem we are able to
prove lower bounds for the general case in Subsection~\ref{lowbound}
and for some special cases in Subsection~\ref{lowboundsp}. Afterwards
we present some upper bounds and give the corresponding algorithms for
the case of abelian groups in Subsection~\ref{upbound}. Here the lower
and upper bounds are quite close together.

Finally in Section~\ref{two} we consider algebraic structures with two
binary operations.  We deal with the situation where we have a ring
with known addition but unknown multiplication.


\section{One binary operation}\label{one}

We define a black box with one binary operation in the most general
way:

\begin{df}
  A \emph{black box groupoid} is a given finite set $S$ together with
  a binary operation $*:S\times S\to S$ which is accessible by an
  oracle.  The oracle can be asked for the multiplication $x*y$ of
  single pairs $(x,y)\in S^2$.
\end{df}

The set $S$ can be thought of as a set of bitstrings and the binary
operation $*$ is the black box we only have limited access to.  

Given a black box groupoid, we are interested in the problem of
recovering the hidden operation $*$ by using a minimal number of
queries to the oracle. We also assume that some information on the
operation $*$ is available, i.e. a set $\mcX$ of possible binary
operations $*:S\times S\to S$ is given. For example, if we know that
$*$ is a group operation, then
\[\mcX_{\text{Groups}}=\{*:S\times S\to S\mid (S,*)\text{ is a group}\}.\]
  
Another example is the situation where we know that $(S,*)$ is isomorphic
to a particular groupoid $(G,\cdot)$. In this case
\begin{multline*}
  \mcX_G=\{*:S\times S\to S\mid\text{there is}\\ \text{an isomorphism
    }f:(S,*)\to(G,\cdot)\}.
\end{multline*}

Algorithms solving the product recovering problem must specify the 
appropriate set $\mcX$.

\begin{rem}
    When only the existence of an epimorphism $f:(S,*)\to (G,\cdot)$ is
    known, then we may have nonunique encoding of group(oid) elements.
    This is usually the case in black box group literature, where there
    is also an oracle for testing whether $f(x)=1$ holds. 
    However, in general we cannot exploit the
  algebraic structure of $G$ to recover exactly $*$.  
  Instead we can hope to find a subset $\tilde{S}\subseteq S$ such that
  $f|_{\tilde{S}}:\tilde{S}\to G$ is bijective and to find
  $\tilde{*}:\tilde{S}\times\tilde{S}\to\tilde{S}$ such that
  $f(a*b)=f(a\tilde{*}b)$ for all $a,b\in S$.
\end{rem}

\subsection{Query-algorithms}

We model query-algorithms as certain labeled trees with the nodes
corresponding to queries to the oracle and the edges corresponding to
its possible answers:

\begin{df} 
  A \emph{query-algorithm} with respect to a set $\mcX$ of binary
  operations $*:S\times S\to S$ on a set $S$ is a rooted tree $T$ with
  labels such that
  \begin{itemize}
  \item any node $v$ of $T$ which is not a leaf is labeled with `$x*y$'
    where $x,y\in S$ (to be thought of as a query), leaves are
    unlabeled,
  \item the branches to the children of $v$ are labeled with `$z$'
    with elements $z\in S$ (to be thought of as possible answers),
    such that different branches have different labels.
  \end{itemize}
  
  Furthermore we require \emph{completeness of answers} in the
  following sense. For every possible binary operation $*\in\mcX$
  there exists a corresponding path $(v_0,\dots,v_k)$ from the root
  $v_0$ to a leaf $v_k$ such that if $v_i$ is labeled with `$x_i*y_i$'
  then the branch $(v_i,v_{i+1})$ is labeled with `$z_i$' where
  $z_i=x_i*y_i$, for $0\leq i<k$.
  
  The leaf $L(*)=:v_k$ is then uniquely determined by this property,
  so that there is a well-defined map \[L:\mcX\to\{\text{leaves of
    }T\}.\] The query-algorithm $T$ is said to \emph{solve
    product-recovering} if this map $L$ is bijective, i.e. there is a
  one-one correspondence between the leaves of $T$ and the operations
  $*\in\mcX$.
\end{df}


\begin{figure}
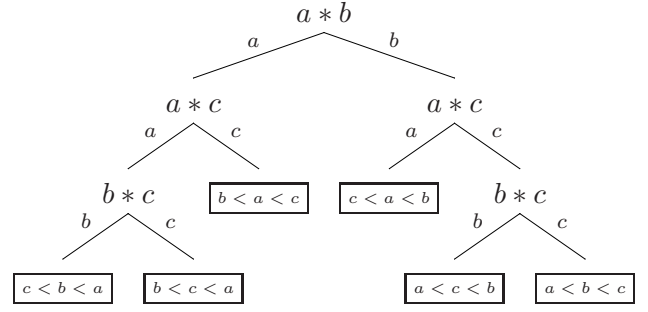

  \centering \Tree[-0.5]{
    & & & & \K{$a*b$}\B{dll}_a\B{drr}^{b} \\
    & & \K{$a*c$}\B{dl}_a\B{dr}^{c} & & & & \K{$a*c$}\B{dl}_a\B{dr}^{c} \\
    & \K{$b*c$}\B{dl}_b\B{dr}^c & & \K{\fbox{\tiny{$b<a<c$}}} &
    & \K{\fbox{\tiny{$c<a<b$}}} & & \K{$b*c$}\B{dl}_b\B{dr}^c \\
    \K{\fbox{\tiny{$c<b<a$}}} & & \K{\fbox{\tiny{$b<c<a$}}} & & &
    & \K{\fbox{\tiny{$a<c<b$}}} & & \K{\fbox{\tiny{$a<b<c$}}}}
   \caption{A query-algorithm for a totally ordered set $\{a,b,c\}$}
   \label{figmax_3}
\end{figure}

\begin{exa}\label{examax_3}
  Let $\mcX$ be the set of all binary operations $*$ on a
  three-element-set $S=\{a,b,c\}$ such that $(S,*)$ is isomorphic to
  the semigroup $(\{0,1,2\},\max)$. Thus, there is an unknown total
  ordering of the elements of $S$, and the problem of
  product-recovering is equivalent to find back this ordering. 
  
  A query-algorithm solving product-recovering is shown in
  Fig.~\ref{figmax_3}. Every leaf $v$ is labeled with the ordering
  which corresponds to the binary operation $*_v=L^{-1}(v)$.
\end{exa}

\begin{figure*}
  \centering \Tree{
    & & & & & \K{$a*a$}\B{dllll}_a\B{dl}^b\B{drr}_c\B{drrrrr}^d \\
    & \K{$b*b$}\B{dl}_a\B{d}_c\B{dr}^d & & &
    \K{$a*b$}\B{dl}_a\B{d}_c\B{dr}^d & & &
    \K{$a*c$}\B{dl}_a\B{d}^b\B{dr}^d & &
    & \K{$a*d$}\B{dl}_a\B{d}^b\B{dr}^c & \\
    \QAnode{a}{b} & \QAnode{a}{c} & \QAnode{a}{d} & \QAnode{b}{a}
    & \QAnode{d}{b} & \QAnode{c}{b} & \QAnode{c}{a} &
    \QAnode{d}{c} & \QAnode{b}{c} & \QAnode{d}{a} & \QAnode{c}{d}
    & \QAnode{b}{d} }
  \caption{A query-algorithm for the group 
    $\Z_4=\{\overline{0},\overline{1},\overline{2},\overline{3}\}$}
  \label{figz_4}
\end{figure*}

\begin{exa}\label{exaz_4}
  Let $(G,\cdot)=(\Z_4,+)$ be the cyclic group of order $4$ and let
  $S=\{a,b,c,d\}$. Fig.~\ref{figz_4} shows a query-algorithm $T$
  solving product-recovering for this group, i.e. with respect to
  $\mcX_G$.
  
  The leaves $v$ are labeled with a shortened presentation of the
  binary operation $*_v=L^{-1}(v)$. Note that the elements
  $\overline{1}$ and $\overline{3}$ of $\Z_4$ are exchangeable, so we
  do not have to specify their corresponding elements in $S$.
 
  Since $T$ has 12 leaves we have
  $|\mcX_G|=12=\frac{4!}{|\Aut(\Z_4)|}$, which also follows from
  Lemma~\ref{lemx_g} below.
\end{exa}


\subsection{Lower bounds}\label{lowbound}

The next lemma establishes a general lower bound.

\begin{lem}\label{lemavg}
  Let $T$ be any query-algorithm which solves product-recovering with
  respect to a set of $\mcX$ of binary operations on a set $S$, and
  let $N$ be its number of queries to the oracle.  Assuming a uniform
  distribution on $\mcX$ we have for the expectation
  \[E(N)\geq\log_{|S|}|\mcX|.\]
  (We say $T$ needs at least $\log_{|S|}|\mcX|$ queries
  \emph{on average}.)
\end{lem}

\begin{proof} 
  $E(N)$ is the average height of all leaves of the tree $T$.  Now any
  node of $T$ has at most $|S|$ children and $T$ has exactly $|\mcX|$
  leaves. This yields the result.
\end{proof}

\begin{lem}\label{lemx_g}
  For any groupoid $(G,\cdot)$ with $|G|=n$ we have
  \[|\mcX_G|=\frac{n!}{|\Aut(G,\cdot)|}.\]
\end{lem}

\begin{proof} 
  Without loss of generality we may assume that $S=G$ as sets.
  
  Consider the set $X$ of all binary operations $*:G\times G\to G$ and
  the group action \[Sym(G)\times X\to X,\quad
  (\varphi,*)\mapsto*^{\varphi},\] where $*^{\varphi}$ is defined by
  $x*^{\varphi}y=\varphi^{-1}(\varphi(x)*\varphi(y))$ for all $x,y\in
  G$, so that $\varphi:(G,*^{\varphi})\to(G,*)$ is a groupoid
  isomorphism.
  
  Now under this group action, the operation $\cdot\in X$, coming from
  the known groupoid $(G,\cdot)$, has exactly $\mcX_G$ as orbit and
  $\Aut(G,\cdot)$ as stabilizer group.  Thus the lemma follows from
  the orbit-stabilizer theorem.
\end{proof}

Now if $G_1,\dots,G_m$ are pairwise non-isomorphic groupoids (e.g. the
family of all abelian groups of a given size), then the $\mcX_{G_i}$
are pairwise disjoint. Let $\mcX:=\mcX_{G_1}\cup\dots\cup\mcX_{G_m}$,
then
\[|\mcX|=n!\left(\sum_{i=1}^m\frac{1}{|\Aut(G_i,\cdot)|}\right).\]

\subsection{Special cases}\label{lowboundsp}

\subsubsection{$\max$-semigroups}

Assume that $(S,*)$ is isomorphic to the semigroup
$(G,\cdot)=(\{0,1,\dots,n-1\},\max)$ (see Example~\ref{examax_3} for
the case $n=3$).

\begin{cor}
  Any query-algorithm which solves product-recovering with respect to
  $\mcX_G$ needs at least \[\log_2(n!)\geq n\log_2n-\frac{n}{\ln
    2}+\frac{\log_2n}{2}\] queries to the oracle on average.
\end{cor}

\begin{proof}
  Since $|\Aut(G,\cdot)|=1$ we have $|\mcX_G|=n!$ by
  Lemma~\ref{lemx_g}. Now in this semigroup any node of a
  query-algorithm has at most $2$ children, so by the same argument
  given in the proof of Lemma~\ref{lemavg} we see that at least
  $\log_2(n!)$ queries on average are needed. Finally the stated
  inequality is a consequence of
  $n!\geq\left(\frac{n}{e}\right)^n\sqrt{n}$, coming from Stirling's
  formula.
\end{proof}

Note that solving product-recovering reduces to the well-studied
problem of sorting an $n$-element set, where the queries to the oracle
correspond to comparisons of elements. There are several sorting
algorithms (e.g. merge sort) which use $O(n\ln n)$ comparisons in the
worst case.

\subsubsection{Abelian groups}

Now assume that $(S,*)\cong(G,\cdot)$ is an abelian group with $n$
elements (see Example~\ref{exaz_4}).

\begin{cor} 
  Suppose that $(G,\cdot)$ is generated by $r$ elements. Any
  query-algorithm which solves product-recovering with respect to
  $\mcX_G$ needs at least
  \[n-\frac{n}{\ln n}+\frac{1}{2}-r\] queries to the oracle on average.
\end{cor}

\begin{proof} 
  Note that any endomorphism of $G$ is determined by its image on its
  $r$ generators. Hence $|\Aut(G)|\leq|\End(G)|\leq n^r$. On the other
  hand we have $n!\geq\left(\frac{n}{e}\right)^n\sqrt{n}$ from
  Stirling's formula, so that $\log_n(n!)\geq n-\frac{n}{\ln
    n}+\frac{1}{2}$. Now the result follows from Lemma~\ref{lemavg}
  and Lemma~\ref{lemx_g}
\end{proof}

Note that any abelian group of size $n$ can be generated by at most
$\log_2n$ elements, so that one can achieve $r\leq\log_2n$ in general.
Of course, if $(G,\cdot)$ is cyclic, one can set $r=1$.

\subsection{Upper bounds for abelian groups}\label{upbound}

We give an upper bound for the worst-case number of queries needed to
solve product-recovering in the case of abelian groups and present the
corresponding algorithm in the proof.

\begin{prop}\label{prop}
  Let $S$ be a set of size $n$ and 
  \[\mcX_{\text{Ab}}
  =\{*:S\times S\to S\mid (S,*)\text{ is an abelian group}\}.\] Then
  there is a query-algorithm which solves product-recovering with
  respect to $\mcX_{\text{Ab}}$ using at most $n$ queries to the
  oracle for any $*\in\mcX$.
\end{prop}

\begin{proof}
  We write the query-algorithm as a list of instructions rather than
  as a tree, because this representation is more compact and readable.
  The algorithm is based on two basic subroutines.
  
  1) Start by choosing some $a=a_1\in S$ and apply the following
  algorithm.
  
  \begin{center}\begin{tabular}{|p{8cm}|}
      \hline 
      Repeat computing $a_{k+1}=a_k*a$ for $k=1,2,3,\dots$
      until $a_{k+1}=a$. \\
      \hline \end{tabular}\end{center}\vspace{1mm}
  
  After execution, $k$ is the order $\ord(a)$ of $a$, and $k$ queries
  to the oracle have been made. We further know that $a_0:=a_k$ is the
  identity element.  Also, we deduce that $a_i*a_j=a_{i+j\mod k}$ for
  $0\leq i,j<k$.
  
  Hence if $S_a=\{a_0,a_1,\dots,a_{k-1}\}$ is the subgroup generated
  by $a$, then $*$ is known on $S_a\times S_a$.
  
  2) If $S_a\neq S$ choose some $b=b_1\in S\setminus S_a$ and apply
  the following algorithm.
  
  \begin{center}\begin{tabular}{|p{8cm}|}
      \hline
      Repeat computing $b_k=b_{k-1}*b$ for $k=2,3,4,\dots$
      until $b_k\in S_a$. \\  
        
      For all $s\in S_a\setminus\{0\}$ and $0<i<k$ compute $s*b^i$. \\
      \hline \end{tabular}\end{center}\vspace{1mm}
  
  After execution we know $s*b^i$ for all $s\in S_a$ and all $0\leq
  i<k$.  Then for any $s,t\in S_a$ and $0\leq i,j<k$ we have by
  commutativity
  \[(s*b_i)*(t*b_j)=
  \begin{cases}(s*t)*b_{i+j}&\text{if }i+j<k,\\
    (s*t*b_k)*b_{i+j-k}&\text{if }i+j\geq k.
  \end{cases}\]
  This element is known, since we knew already $*$ on $S_a\times S_a$.
  It follows that $*$ is known on $S_{ab}\times S_{ab}$ where $S_{ab}$
  is the subgroup generated by $S_a$ and $b$.
  
  Let $m=|S_a|$. The number of queries to the oracle needed by the
  algorithm is
  \[k-1+(m-1)(k-1)=m(k-1)=mk-m.\] 
  Now $mk=|S_{ab}|$, so that $|S_{ab}|-|S_a|$ queries to the oracle
  have been used.
  
  3) If $S_{ab}\neq S$ choose some $c\in S\setminus S_{ab}$ and repeat
  2) with $S_a$ replaced by $S_{ab}$ and $b$ replaced by $c$.  After
  that $*$ is known on $S_{abc}\times S_{abc}$, the subgroup generated
  by $S_{ab}$ and $c$, and $|S_{abc}|-|S_{ab}|$ queries to the oracle
  have been used, etc.
  
  Writing $S_1,S_2,S_3,\dots$ for $S_a,S_{ab},S_{abc},\dots$ we
  finally reach $r$ such that $S_r=S$.  Then we have recovered the
  whole operation $*$ on $S\times S$ and we have used
  \[|S_1|+(|S_2|-|S_1|)+\dots+(|S_r|-|S_{r-1}|)=|S_r|=|S|=n\]
  queries to the oracle in total.
\end{proof}

\begin{exa}
  Consider a black box group $(S,*)$ of size $11$.  Then we know that
  $S$ is isomorphic to the cyclic group and $|\Aut(S,*)|=10$.  By
  Lemma~\ref{lemavg} and Lemma~\ref{lemx_g} we conclude that an
  algorithm solving product-recovering needs at least
  $\lceil\log_{11}3991680\rceil=7$ queries to the oracle in the worst
  case.
    
  Proposition~\ref{prop} ensures the existence of an algorithm which
  needs $11$ oracle-queries. In fact it is not hard to see that in the
  case of groups of prime order the last two queries in the above
  algorithm can be omitted, yielding an algorithm which uses $9$
  queries to the oracle.
    
  Now a computer search among all possible product recovering
  algorithms has shown that the minimal number of oracle-queries an
  algorithm needs in the worst-case is~$8$.  Such an algorithm can be
  outlined as follows:
  \begin{enumerate}
  \item choose some $a\in S$ and compute $a_{\text{sq}}=a*a$
  \item if $a_{\text{sq}}\neq a$ let $a_1=a$ and $a_2=a_{\text{sq}}$,
    otherwise let $e=a$, choose some $a_1\neq a$ and compute $a_2=a*a$
  \item compute $a_3=a_2*a_1$, $a_4=a_3*a_1$, $a_5=a_4*a_1$ and
    $a_7=a_5*a_2$
  \item if $a_{\text{sq}}\neq a$ compute $e=a_7*a_4$
  \item choose three different elements $b,c,d$ from the
    four-element-set $S\setminus\{e,a_1,a_2,a_3,a_4,a_5,a_7\}$ and
    compute $b*c$ and $b*d$
  \end{enumerate}
\end{exa}

\begin{rem}
  When we have nonunique encoding of group elements a variant of the
  above algorithm will solve product-recovering provided we are given
  a generating set and an oracle for recognizing the identity.  However,
  to check whether $b_k$ lies in $S_a$ (as in the second subroutine)
  may be a costly operation.
\end{rem}


\section{Two binary operations}\label{two}

Suppose we are given a finite set $S$ with \emph{two} hidden binary
operations $+$ and $*$, accessible via an oracle. If we know that
$(S,+,*)$ is a ring, then $(S,+)$ is an abelian group, so we can use
Proposition~\ref{prop} to recover the addition table. We now have a
ring with known addition, but unknown multiplication. This section
deals with that situation.

\begin{df}
  A \emph{black box groupoid with given addition} is a black box
  groupoid $(S,*)$ such that there is a known binary operation
  $+:S\times S\to S$ on $S$, and the following distributive laws hold
  on $S$ with respect to $+$ and $*$, i.e.
  \[\left.\begin{array}{c}
  a*(b+c)=(a*b)+(a*c)\\
  (a+b)*c=(a*c)+(b*c)
  \end{array}\right\}\text{ for all }a,b,c\in S.\]
\end{df}

In this case, all binary operations $*\in\mcX$ in question will
satisfy the distributive laws above. Suppose, for example, we know
that $(S,+,*)$ is isomorphic to some known ring $(R,+,\cdot)$. Then
the set of possible operations we are dealing with is
\begin{multline*}
\mcX_R:=\{*:S\times S\to S\mid\text{there is an}\\ \text{isomorphism
  }\varphi:(S,+,*)\to(R,+,\cdot)\}.
\end{multline*}

Its size is given in the next result.

\begin{lem}\label{lemx_r}
  For any ring $(R,+,\cdot)$ we have
  \[|\mcX_R|=\frac{|\Aut(R,+)|}{|\Aut(R,+,\cdot)|},\]
  where $\Aut(R,+)$ are the additive group automorphisms and $\Aut(R,+,\cdot)$
  are the ring automorphisms.
\end{lem}

\begin{proof}
  The arguments are the same as in the proof of Lemma~\ref{lemx_g}. We
  identify $S=R$ as sets and consider the action of $\Aut(R,+)$ on the
  set $X$ of all binary operations $*:R\times R\to R$. Then $\mcX_R$
  is exactly the orbit of $\cdot\in X$ and $\Aut(R,+,\cdot)$ is its
  stabilizer group.
\end{proof}

Now we specialize to the case when $(S,+,*)\cong(\F_q,+,\cdot)$ is a
field of size $q=p^r$ with $p$ prime.

\enlargethispage{.5cm}

\begin{cor}
  Any query-algorithm which solves product-recovering for a field of
  size $q=p^r$ with known addition needs at least \[r-\log_q(4r)\]
  queries to the oracle on average.
\end{cor}

\begin{proof}
  The automorphisms $\Aut(\F_q,+)$ are exactly the vector space
  automorphisms of the $\F_p$-vector space $(\F_p)^r$, so that
  \begin{multline*}
    |\Aut(\F_q,+)|=(q-1)(q-p)\cdots(q-p^{r-1})\\
    =q^r\left(1-\frac{1}{p}\right)\left(1-\frac{1}{p^2}\right)
      \cdots\left(1-\frac{1}{p^r}\right),
  \end{multline*}
  where $(1-\frac{1}{p})\cdots(1-\frac{1}{p^r})> 
  \prod\limits_{i\geq 1}(1-\frac{1}{2^i})>
  \prod\limits_{i\geq 0}(\frac{1}{2})^{\frac{1}{2^i}}=\frac{1}{4}$.
  On the other hand, $|\Aut(\F_q,+,\cdot)|=|\Aut(\F_{p^r}/\F_p)|=
  [\F_{p^r}:\F_p]=r$, by basic Galois theory. Hence
  \[\frac{|\Aut(\F_q,+)|}{|\Aut(\F_q,+,\cdot)|}\geq\frac{q^r}{4r}\]
  and the result follows from Lemma~\ref{lemavg} and Lemma~\ref{lemx_r}.
\end{proof}

We now give an upper bound for the number of queries needed to solve
product-recovering for rings $(S,+,*)$ of size $|S|=n$ with given
addition.  For this it suffices to ask the oracle for all products
$a*b$ of elements $a,b\in A$, where $A$ is a generating set for the
abelian group $(S,+)$. Then if $x,y\in S$, we can write
$x=a_1+\dots+a_k$ and $y=b_1+\dots+b_l$ with $a_i,b_j\in A$ for all
$i,j$, and thus
\[x*y=(a_1+\dots+a_k)*(b_1+\dots+b_l)=\sum_{i=1}^k\sum_{j=1}^la_i*b_j;\]
now, since all $a_i*b_j$ are known and the addition is known, we also know
$x*y$.

Because $(S,+)$ can be generated by at most $\log_2n$ elements, we
thus have $(\log_2n)^2$ as an upper bound. Together with
Proposition~\ref{prop} this proves:

\begin{prop}
  If $(S,+,*)$ is a ring of size $n$, then there is a query-algorithm
  which solves product-recovering for both operations $+$ and $*$ with
  at most
  \[n+(\log_2n)^2\] queries to the oracle in the worst case.
\end{prop}

\end{document}